\documentclass[conference, a4paper, twocolumn]{IEEEtran}
\usepackage{xcolor}
\usepackage{graphicx, epsfig, amssymb, amsmath,amsthm}
\usepackage{multirow}
\usepackage{slashbox, pict2e}
\usepackage{times}
\usepackage{cases}
\usepackage{array}
\usepackage{float}
\usepackage{algorithm}
\usepackage{algorithmic}
\usepackage{subfigure}
\newtheorem{theorem}{Theorem}
\usepackage{cite}
\usepackage{bm}
\usepackage{setspace}
\usepackage{flushend}
\IEEEoverridecommandlockouts
\usepackage{longtable}
\usepackage{makecell}

\begin{document}
\title{A Physics-based and Data-driven Approach  for Localized  Statistical Channel Modeling}

\bibliographystyle{IEEEtran}

\author{
	\IEEEauthorblockN{Shutao~Zhang\IEEEauthorrefmark{1}\IEEEauthorrefmark{2}, Xinzhi~Ning\IEEEauthorrefmark{3}\IEEEauthorrefmark{4}, Xi~Zheng
		\IEEEauthorrefmark{5}, Qingjiang~Shi\IEEEauthorrefmark{3}\IEEEauthorrefmark{2}, Tsung-Hui~Chang\IEEEauthorrefmark{1}\IEEEauthorrefmark{2}, and
	    Zhi-Quan~Luo\IEEEauthorrefmark{1}\IEEEauthorrefmark{2}\IEEEauthorrefmark{4}}
	\IEEEauthorblockA{\IEEEauthorrefmark{1}School of Science and Engineering, The Chinese University of Hong Kong, Shenzhen, China}
	\IEEEauthorblockA{\IEEEauthorrefmark{2}Shenzhen Research Institute of Big Data, Shenzhen, China}
	\IEEEauthorblockA{\IEEEauthorrefmark{3}School of Software Engineering, Tongji University, Shanghai, China}
	\IEEEauthorblockA{\IEEEauthorrefmark{4}Pengcheng Lab, Shenzhen, China}
	\IEEEauthorblockA{\IEEEauthorrefmark{5}Networking and User Experience Lab, Huawei Technologies, Shenzhen, China}
	Email: {shutaozhang@link.cuhk.edu.cn}
	
	\thanks{The work is supported in part by the Shenzhen Science and Technology Program under Grant No. RCJC20210609104448114, and in part by the NSFC, China, under Grant No. 61731018 and 62071409, and by Guangdong Provincial Key Laboratory of Big Data Computing.}

}

\maketitle
\thispagestyle{empty}

\begin{abstract}
\noindent Localized channel modeling is crucial for offline performance optimization of 5G cellular networks, but the existing channel models are for general scenarios and do not capture local geographical structures. In this paper, we propose a novel physics-based and data-driven localized statistical channel modeling (LSCM), which is capable of sensing the physical geographical structures of the targeted cellular environment. The proposed channel modeling solely relies on the reference signal receiving power (RSRP) of the user equipment, unlike the traditional methods which use full channel impulse response matrices. The key is to build the relationship between the RSRP and the channel's angular power spectrum. Based on it, we formulate the task of channel modeling as a sparse recovery problem where the non-zero entries of the sparse vector indicate the channel paths' powers and angles of departure. A computationally efficient weighted non-negative orthogonal matching pursuit (WNOMP) algorithm is devised for solving the formulated problem.   Finally, experiments based on synthetic and real RSRP measurements are presented to examine the performance of the proposed method.  

\end{abstract}


%
%
\section{Introduction}
\label{Section-Introduction}
With the rapid development of the fifth generation (5G) mobile communications, cellular network optimization becomes the fundamental means to improve the quality of experience (QoE) of the users in wireless communication \cite{liu2019efficient, li2021zeroth, 9414155}. Since network optimization must be performed in the offline fashion, one of the key challenges  is how to characterize the large-scale behaviors of the wireless propagation channels between the base station and the user equipment under the network coverage. Thus, a  channel model that can capture the localized physical geographical structures of the targeted environment is paramount  for the  network optimization of the 5G systems.

Unfortunately,  the existing channel models commonly used in  the 5G systems are not applicable to  cellular network optimization.  For example,  the geometry-based stochastic models (GBSM) have been intensively studied \cite{wu2017general, wang2018survey, bian2021general}.  They are employed for some typical scenarios, such as indoor and outdoor, rural and urban scenarios, but lack the ability to describe a specific and localized environment. As for deterministic channel modeling methods like ray tracing \cite{hussain2019efficient}, they can describe the localized environment, but the complexity of using Maxwell's equations is unaffordable for cellular network optimization. Additionally, ray tracing is based on map information, which may not always be available \cite{rodriguez2020network}. 

In this paper, we  propose a physics-based and data-driven localized statistical channel modeling (LSCM) method, which is capable of sensing the physical propagation environment for the target region. The  proposed LSCM  can efficiently provide the statistical information of channel angular power spectrum (APS),  thus revealing the localized multi-path structure and topography.  In contrast to the existing GBSM which relies on channel impulse response (CIR) measurements, our proposed LSCM considers the  measurements of reference signal receiving power (RSRP) \cite{park2016analysis} from multiple transmit beams only.  Compared with the high-dimensional CIR, the benefit of using RSRP is that it can save a lot of storage space during data collection, and enable computationally efficient approaches. However, how to extract the channel statistics from the RSRP data has never been studied before and is a challenging task. 

To this end, we first build the relationship between the RSRP and the APS.  We  show later that extracting the  APS can be achieved through solving a sparse recovery problem,  where the non-zero entries of the sparse vector indicate the channel paths' powers and angles of departure. While the classical compressive sensing (CS) theory suggests that orthogonal matching pursuit (OMP) is a good solver \cite{tropp2007signal}, its convergence is guaranteed  for some well-conditioned matrices, which in general does not hold.   
On the contrary, the column magnitudes  of the coefficient matrix in the formulated problem can vary dramatically, which will significantly deteriorate the channel modeling performance. To solve this intractable sparse recovery problem,  we design the weighted non-negative OMP (WNOMP) algorithm. Simulation results based on synthetic data and real-world RSRP measurements  demonstrate that the proposed channel modeling method is effective and the WNOMP outperforms the classical methods. 


\begin{figure}[t]
	\centering
	\includegraphics[width=3.1in, clip, keepaspectratio]{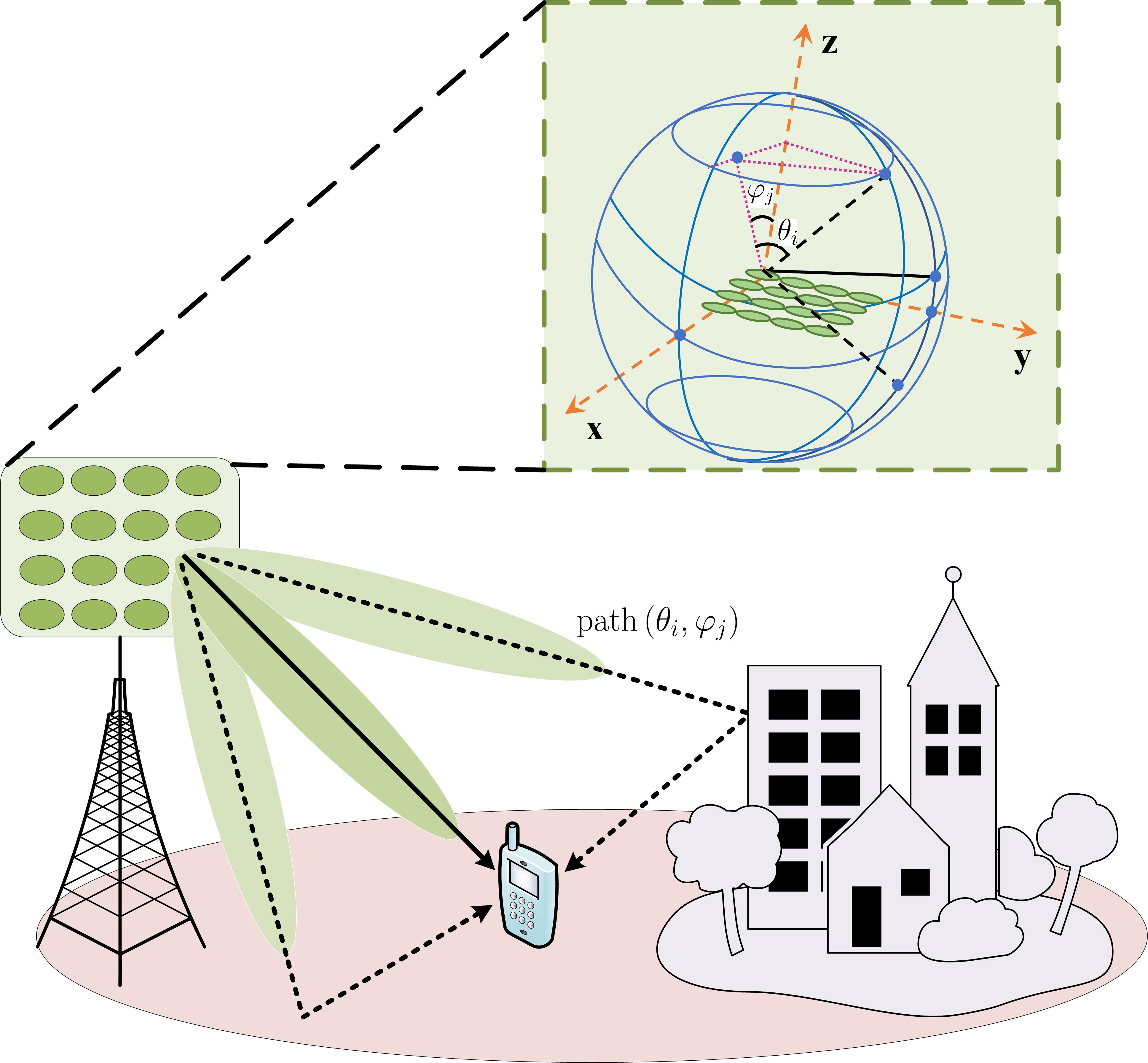}
	\vspace{-2pt}
	\caption{Angular discretization of the channel model in the downlink.}
	\vspace{-2pt}
	\label{fig_1}
\end{figure}

\section{MIMO Downlink System and RSRP}
\label{Section-OVERVIEW}

In this section, we consider the channel modeling of the massive multiple input multiple output (MIMO) downlink communication system in the target grid,  followed by the RSRP measured from multiple beams. 

\subsection{Angular discretization channel modeling}
The technique of beamforming is exploited in the massive MIMO of 5G systems, yielding high-resolution multi-path channels in the angular domain. Thus, we consider the channel modeling of the massive MIMO downlink communication system with beamforming, where the signal is transmitted from a base station to the user equipment through several propagation paths, as shown in Fig.~\ref{fig_1}.   

Similar to the 3GPP's technical report  \cite{38901},  we focus on the  tilt angle of departure (AoD),   azimuth AoD and channel gain from the base station to the user equipment with one receive antenna,  ignoring the arrival angle and delay.  Suppose the uniform rectangular  array of the base station  contains $N_T = N_x \times N_y$ antennas. The CIR  of antenna $(x, y)$ from the base station to the user equipment is given by    
\begin{align} \label{eq:hxy}
	h_{x, y} (t) = \sum_{i = 1}^{N_V} \sum_{j = 1}^{N_H} & \sqrt{\alpha_{i, j}(t) } \times g_{i, j} \times e^{-j2\pi\frac{d_x x}{\lambda}\cos \theta_i\sin \varphi_j}  \nonumber\\ & \times e^{-j2\pi\frac{d_y y }{\lambda}\sin \theta_i } \times e^{-j\omega_{i,j} (t)-j\omega_{x,y} (t)}.
\end{align}

\begin{table}[t]
	\caption{Summary of notation}
	\vspace{-1pt}
	\centering
	\includegraphics[width=3.42in, clip, keepaspectratio]{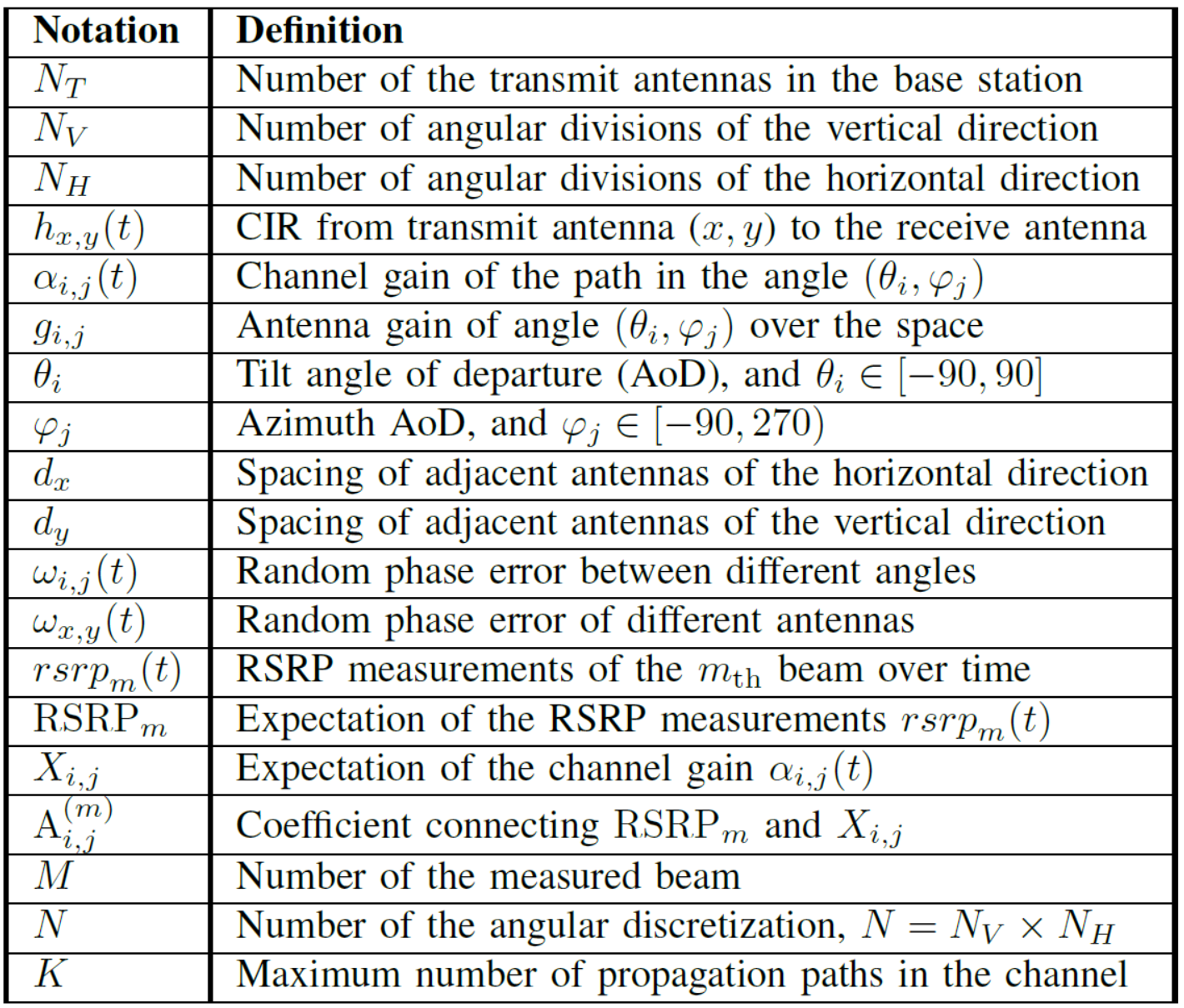}
	\label{parameter22} 
\end{table}

The definition of notation in \eqref{eq:hxy} is summarized in Table~\ref{parameter22}. The tilt AoD and azimuth AoD over the free space are discretized into $N_V$ angles and $N_H$ angles, respectively. If there does not exist a path in the angles $\left( \theta_i, \varphi_j  \right)$, the corresponding channel gain $\alpha_{i, j}(t) $ is zero, otherwise $\alpha_{i, j}(t) > 0$  for the path $\left( \theta_i, \varphi_j  \right)$ in the channel.  It is readily known that the channel gain ${\alpha_{i, j} }(t) $ is sparse, and contains multiple propagation paths in the angular domain.

The channel gain ${\alpha_{i, j} }(t) $ consists of path loss and the shadowing effect. Path loss is determined by the physical environment (distance, carrier frequency, buildings) which is assumed to be relatively static, while shadowing is caused by the obstacles and is usually modeled as the log-normal distribution.  Thus, we assume ${\alpha_{i, j} }(t) $  follows the log-normal distribution, with its mean representing  path loss and its covariance representing the shadowing effect. $\omega_{i,j} (t)$ is the random phase error between different angles caused by the  reflection, diffraction  and scattering effect of  electromagnetic waves, while $\omega_{x,y} (t)$ is the random phase error of different antennas caused by the imperfect hardware of the antenna array. Inspired by \cite{38901}, we assume  $\omega_{i,j} (t)$ follows the uniform distribution  between $-\pi$ and $\pi$, and $\omega_{x,y} (t)$ follows the Gaussian distribution with zero mean and covariance $\sigma^2$.

\begin{figure*}[!t]
	\centering
	\subfigure[RSRP = -75.92 dB]{\includegraphics[width=1.3in]{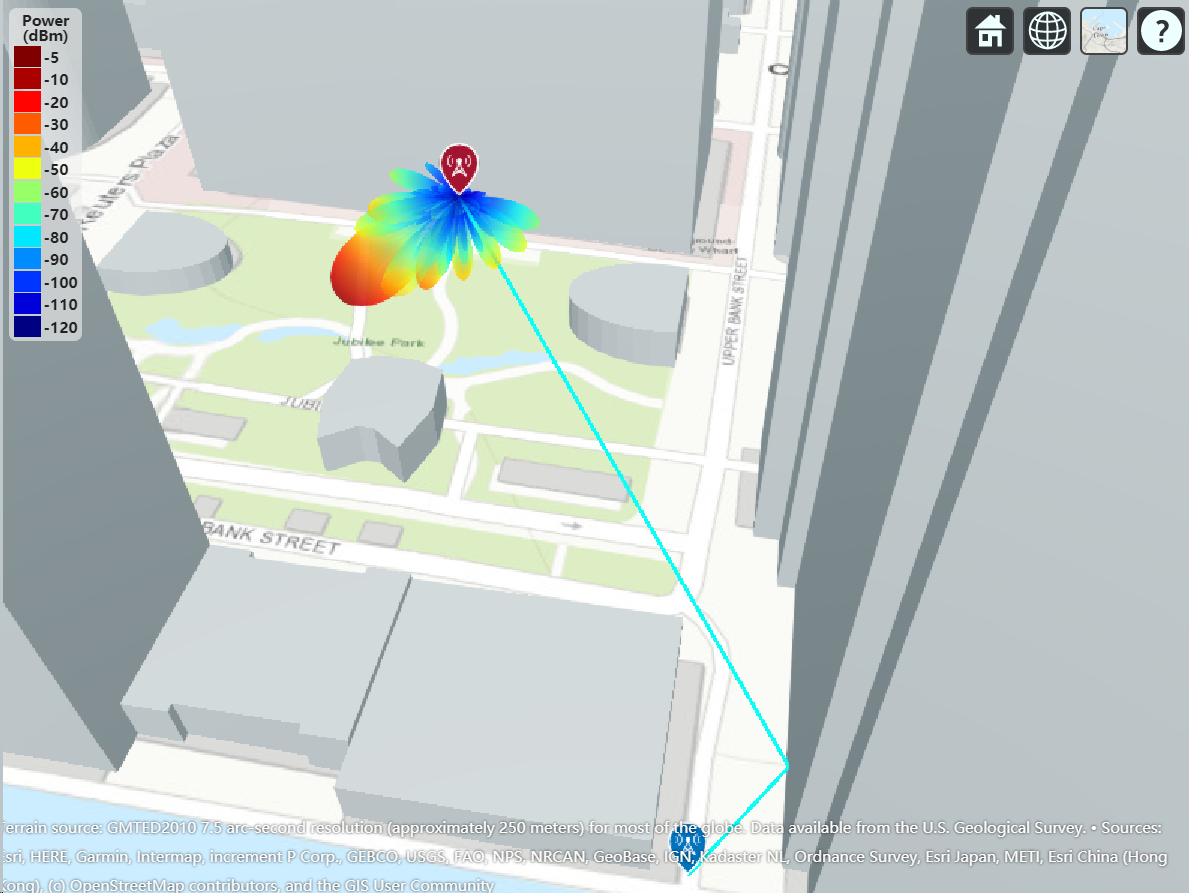}}
	\subfigure[RSRP = -73.29 dB]{\includegraphics[width=1.3in]{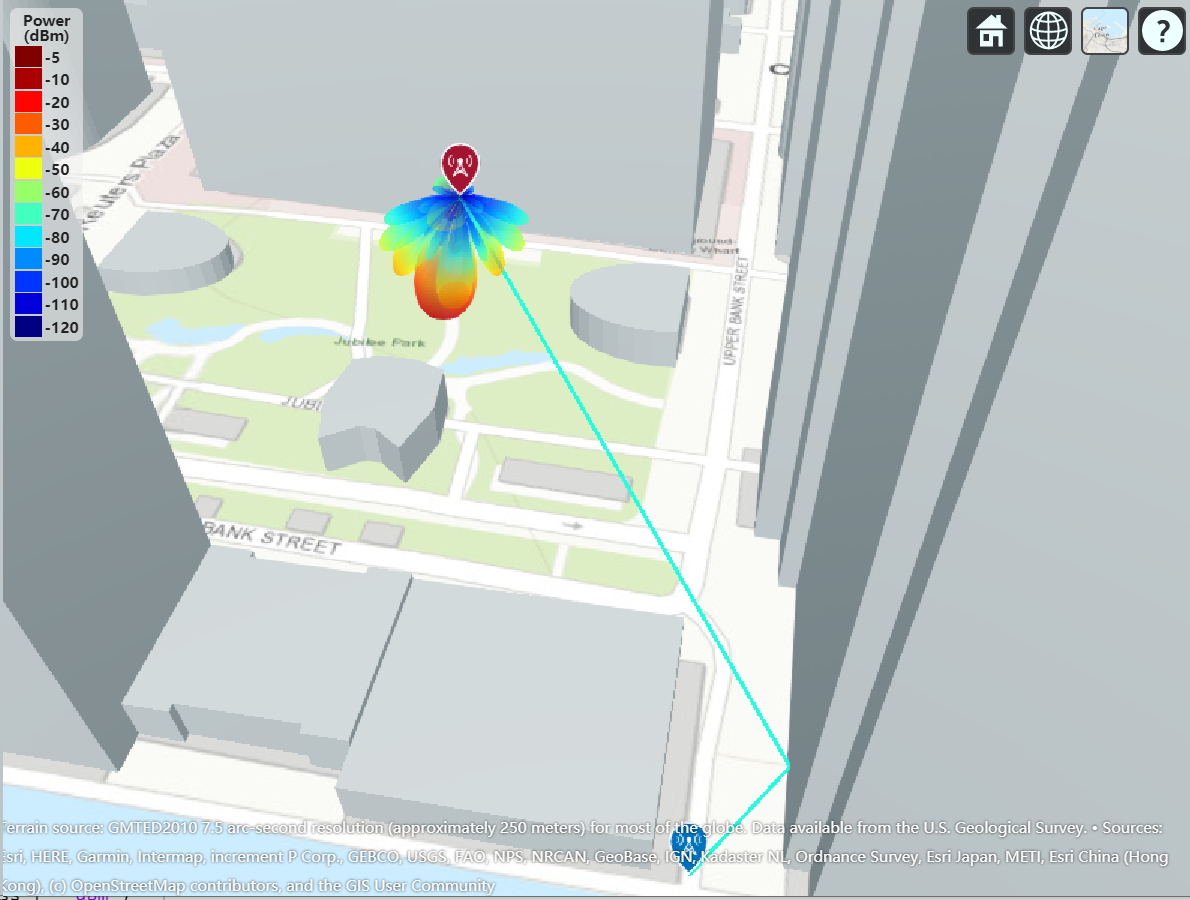}}
	\subfigure[RSRP = -64.12 dB]{\includegraphics[width=1.3in]{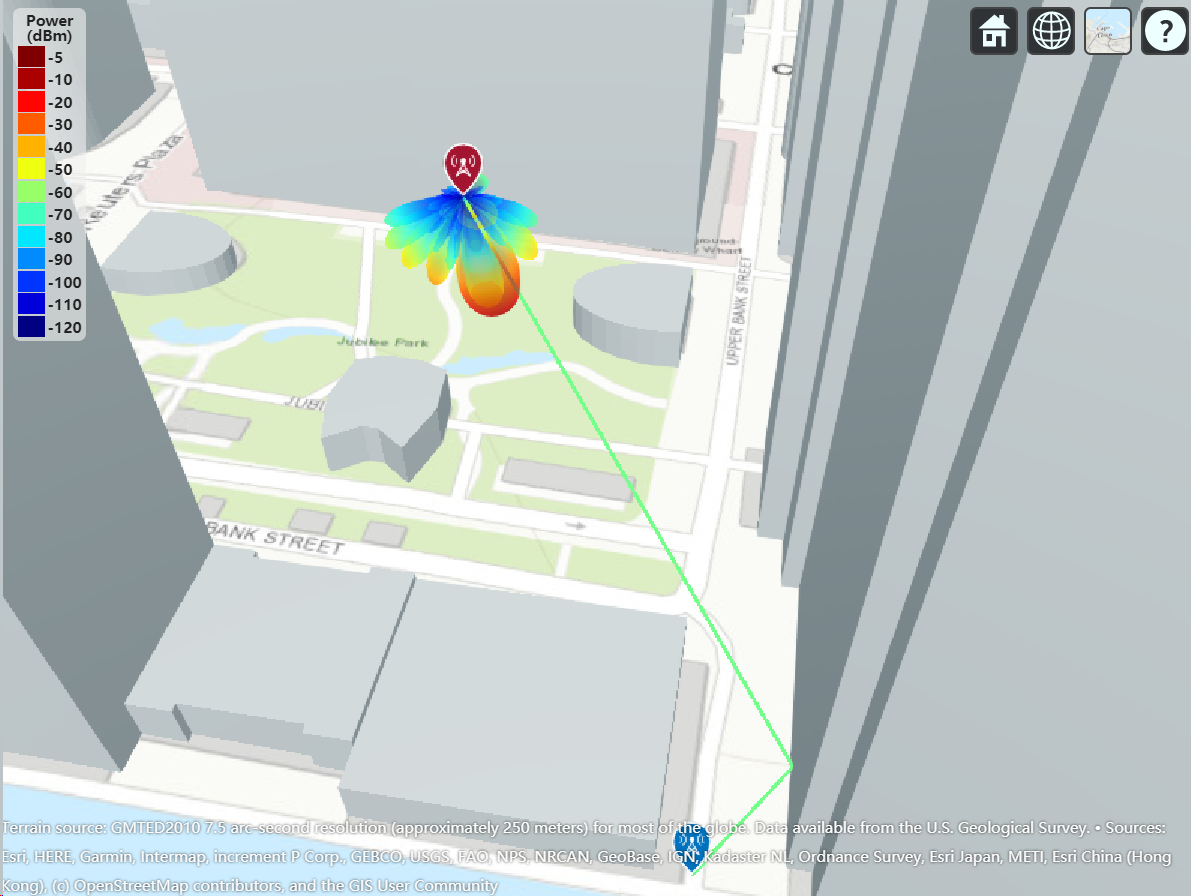}}
	\subfigure[{\bf RSRP = -59.61 dB}]{\includegraphics[width=1.3in]{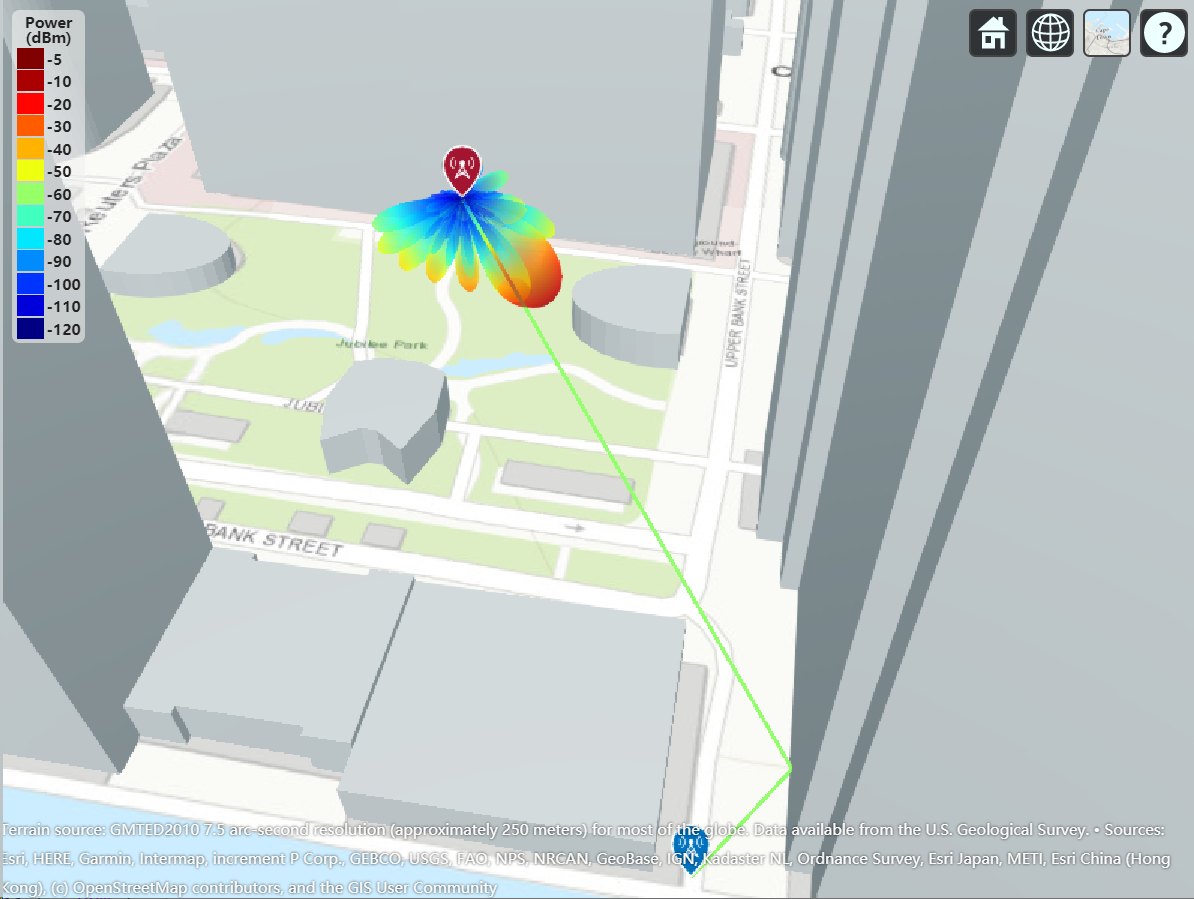}}
	\subfigure[RSRP = -70.73 dB]{\includegraphics[width=1.3in]{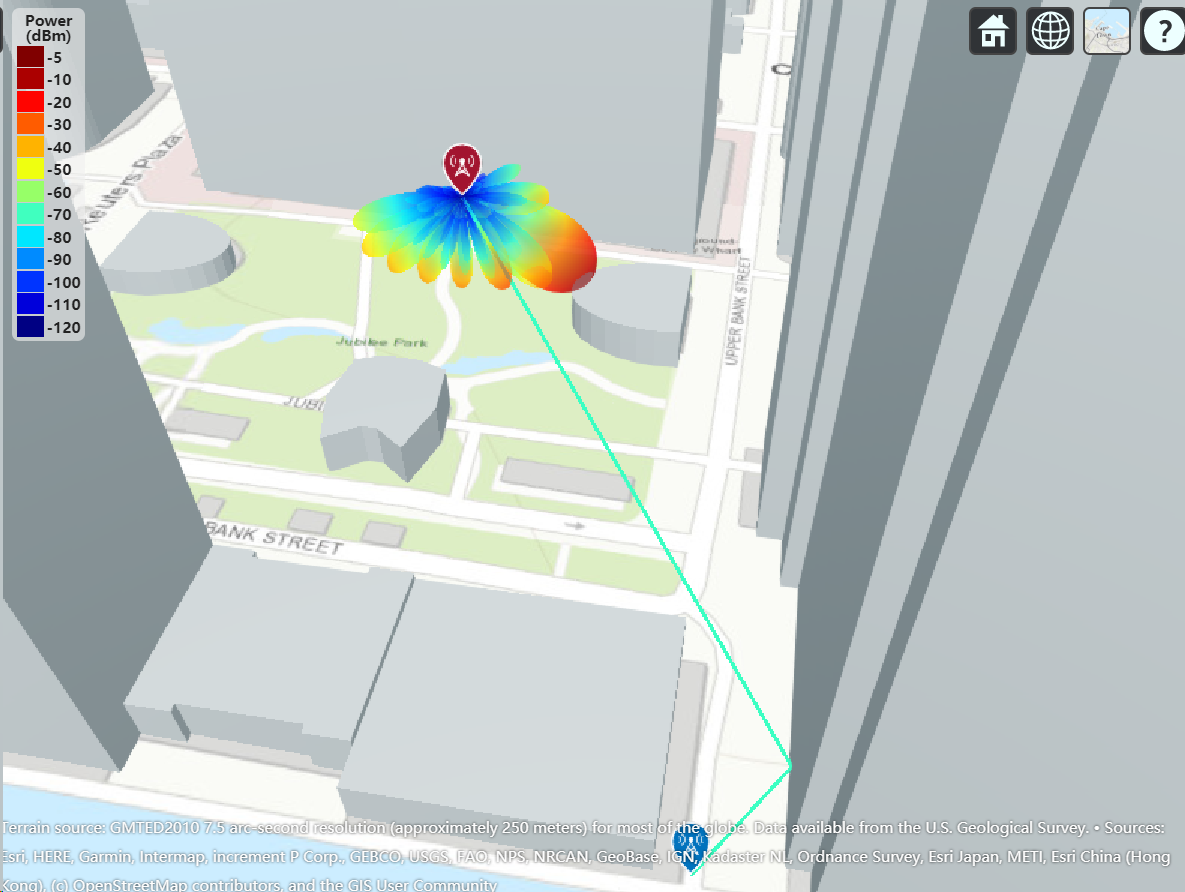}}
	\vspace{-1pt}
	\caption{An example of channel measurements of the RSRP  from five beams. There exists one propagation path between the base station and target grid, with angle of departure $\left[\widehat{\theta}, \widehat{\varphi} \right] = [-2.84, 23.82]$. The directions of main lobe of different beams are: (a)[0, -15], (b)[0, 0], (c)[0, 15], (d)[0, 30], (e)[0, 45].}
	\vspace{-1pt}
	\label{fig_0}
\end{figure*}

\subsection{RSRP measurements}

The channel measurement of our proposed LSCM is the  RSRP measured from multiple beams. In the downlink of 5G cellular systems, the RSRP can be measured from synchronization signals block (SSB) beams or channel state information-reference signal (CSI-RS) beams. Denote the precoding matrix of the $m_{\rm th}$  beam as $\boldsymbol{W}^{(m)} \in \mathbb{C}^{N_x \times N_y}$, whose size is the same as that of the antenna array. The entry of $\boldsymbol{W}^{(m)}$ is  $ \left(w_{x, y}^{(m)}\right) = \left(e^{j\phi_{x,y}^{(m)}}\right)$.   
Denote the CIR matrix from the base station to the user equipment as $\boldsymbol{H}  \in \mathbb{C}^{N_x \times N_y}$, where the entry  of $\boldsymbol{H} $ is $ h_{x, y} (t)$ in \eqref{eq:hxy}. The RSRP of the $m_{\rm th}$  beam  at time $t$ is defined as
\begin{equation} \label{eq: rsrp}
	{rsrp}_{m}(t) = P\left|\operatorname{tr}\left(\boldsymbol{H}^{{T}} \boldsymbol{W}^{(m)}\right)\right|^{2} = P\left|\sum_{x, y} h_{x, y} (t) w_{x, y}^{(m)}\right|^{2},
\end{equation}
where $P$ denotes the transmit power.   The quality of the channel can be represented through  ${rsrp}_{m}(t)$, and a larger value of ${rsrp}_{m}(t)$ reflects  better channel quality.

Figure~\ref{fig_0} shows an example  of channel measurement of RSRP measured from five beams. These five beam patterns are generated according to different beamforming weights, and the directions of main lobe of these five beams are different. The receiver at the target grid is highlighted as a blue point.  The target grid can receive five RSRP measurements, corresponding to the received power of five beams. To be specific,  the strongest value of RSRP measurements is -59.61 dB in Fig.~\ref{fig_0}(d), because its main lobe's direction is $\left[ \bar{\theta}, \bar{\varphi}   \right] = [0, 30]$, which is close to the angle of departure of the propagation path $\left[\widehat{\theta}, \widehat{\varphi} \right] = [-2.84, 23.82]$.  

\section{Physics-based and Data-driven LSCM}
\label{Section-Simulation}
In this section, the relation between the first-order statistics of the RSRP and channel gain is  established, which is the key of the proposed LSCM. Then, the WNOMP  is designed for extracting the channel statistics of the APS from the RSRP. 

\subsection{The first-order statistics of the RSRP and channel gain}

As for  ${rsrp}_{m}(t)$, it is  a random variable  inferred from \eqref{eq: rsrp}, since $ h_{x, y} (t)$ contains three independent random variables, i.e., $\omega_{x,y} (t)$, $\omega_{i,j} (t)$ and ${\alpha_{i, j} }(t) $.  Theorem~\ref{theorem1} shows the relationship between the first-order statistics of ${rsrp}_{m}(t)$ and ${\alpha_{i, j} }(t) $, i.e., their expectation  $\mathbb{E} \left( {rsrp}_{m}(t) \right)$ and $\mathbb{E} \left( {\alpha_{i, j} }(t)  \right)$, implying the statistical relationship of the RSRP and angular information.

\begin{theorem} \label{theorem1}
	{\rm{(The relationship between the first-order statistics of RSRP and channel gain)}} 
	Suppose the expectation of ${rsrp}_{m}(t)$ is   ${\rm{RSRP}}_{m} \triangleq \mathbb{E}\left( {rsrp}_{m}(t) \right)$, and the expectation of channel gain is $X_{i,j}  \triangleq \mathbb{E}\left( {\alpha_{i, j} }(t)  \right)$, then we have 
	\begin{align}\label{eq11}
		{\rm RSRP}_{m} = \sum_{i = 1}^{N_V} \sum_{j = 1}^{N_H}  {\rm A}_{i, j}^{(m)}X_{i,j} ,
	\end{align}
	where  
	\begin{align} \label{eq: 886}
		{\rm A}_{i, j}^{(m)}   & \triangleq P g_{i, j}^{2}\Bigg( \! N_{x} N_{y}\left(1-e^{-\sigma^{2}}\right) \nonumber \\ 
		& \quad  \  + e^{-\sigma^{2}} \sum_{x,y}  \sum_{x^{\prime},y^{\prime}} \cos \left(\psi_{i, j, x, y}^{(m)}-\psi_{i, j, x^{\prime}, y^{\prime}}^{(m)} \! \right) \! \Bigg), \\
		\psi^{(m)}_{i,j,x,y} &= 2\pi {d_x x \over \lambda}\cos \theta_i \sin \varphi_j + 2\pi {d_y y \over \lambda  } \sin \theta_i - \phi_{x,y}^{(m)}.
	\end{align}
\end{theorem}

 \begin{proof}
	Please refer to Appendix~\ref{appendix1}.
\end{proof}

The statistical relationship between the expectation of RSRP measurements ${\rm{RSRP}}_{m}$ and the first-order statistics of channel gain $X_{i,j} $ is connected by the coefficient ${\rm A}_{i, j}^{(m)}$, as revealed by  Theorem~\ref{theorem1}.  According to this finding, it is possible to extract the angular information from RSRP measurements if we can calculate $X_{i,j} $ from ${\rm{RSRP}}_{m}$. The APS can be expressed by $X_{i,j} $, where the subscripts $i$ and $j$ indicate the tilt angle $\theta_i$ and azimuth angle $\varphi_j$, respectively. 

Based on \eqref{eq11}, we reshape  $X_{i,j} $ with the subscripts $i = 1, \cdots, N_V, j = 1, \cdots, N_H$ as a vector, i.e., 
\begin{align} \label{20220424}
	{\bf x}  = \left[ X_{1,1} , X_{1,2} , \cdots, X_{1,N_H} , X_{2,1} , \cdots, X_{N_V,N_H}  \right]^T,
\end{align} 
where the length of ${\bf x} $ is $N = N_V \times N_H$.  If the path $\left( \theta_i, \varphi_j \right)$ does not exist, then   $\mathbb{E} \left( {\alpha_{i, j} }(t)  \right) = 0$.  As there are only several propagation paths  from the base station to the user equipment, the channel vector ${\bf x} $ is sparse.

Suppose there are $M$ measured beams in total,  and define ${\bf y}  \in \mathbb{R}^{M \times 1}$ as the vector of the expectation of RSRP measurements from different beams, i.e., 
\begin{align} 
	{\bf y}  = \left[{\rm RSRP}_{1}, {\rm RSRP}_{2}, \cdots, {\rm RSRP}_{m}, \cdots, {\rm RSRP}_{M}  \right]^T.
\end{align}

Define the  vector ${\bf a}^{(m)} \in \mathbb{R}^{N \times 1}$ as 
{\small
\begin{align} \label{rowvector}
	{\bf a}^{(m)} = \left[  {\rm A}_{1, 1}^{(m)},  {\rm A}_{1, 2}^{(m)}, \cdots,  {\rm A}_{1, N_H}^{(m)},  {\rm A}_{2, 1}^{(m)}, \cdots,  {\rm A}_{N_V, N_H}^{(m)} \right]^T,
\end{align}
}
and the coefficient matrix ${\bf A} \in \mathbb{R}^{M \times N}$ is
\begin{align}
	{\bf A} = \left[ {\bf a}^{(1)}, {\bf a}^{(2)}, \cdots,  {\bf a}^{(m)}, \cdots,  {\bf a}^{(M)} \right]^T.
\end{align}

Then we can recast \eqref{eq11} as 
\begin{align}	 \label{eq: problem1}
	{\bf y} = {\bf A} {\bf x}.
\end{align}


%

%
\subsection{The designed WNOMP  algorithm}
To construct the localized statistical channel model, we want to recover the sparse channel vector $\bf x$ from the expectation of RSRP measurements $\bf y$ and the coefficient matrix $\bf A$.  According to \eqref{eq: problem1}, as the expectation of channel gain in terms of different angles  ${ x}_{n} \geq 0,  \ n = 1, \cdots, N$, is non-negative, we can formulate the optimization problem  as
\begin{equation}\label{problem2}
	\begin{aligned}
		\ \min _{\bf x} & \  \|{\bf {Ax-y}}\|_{2}^{2}\\
		\text {s.t.} & \ \|{\bf x}\|_{0} \leq K, \\
		&\ {x}_{n} \geq 0,  \ n = 1, \cdots, N,
	\end{aligned}
\end{equation}
where $K$ is the maximum number of non-zero entries, representing the maximum number of propagation paths.

The OMP is an iterative algorithm, which selects one column into the set of indexes of non-zero entries $\mathcal{S}$ at each iteration. The selection criteria is to select the column that is best correlated with the residual of current iteration. Problem \eqref{problem2} can be solved by the non-negative OMP (NNOMP),  as shown in Algorithm~\ref{NNOMP}. The biggest difference between NNOMP and  OMP lies in  Step  5 of Algorithm~\ref{NNOMP}, where the subproblem of non-negative least squares (NNLS) is solved instead of the unconstrained least squares (ULS). After solving the NNLS subproblem, the entries of $\bf x$ in the complement set $\mathcal{S}^c$ are forced  zero to ensure the sparsity. A fast implementation of NNOMP is proposed in \cite{nguyen2019non} based on the active-set algorithm. The role of $ \max \left( {\bf A}^{T} {\bf r}_{k} \right) < 0$  is also explained in \cite{nguyen2019non} by using the Karush-Kuhn-Tucker (KKT) condition.  The condition ${\bf a}_n^T{\bf r}_k > 0$ means that the $n_{\rm th}$ column is the descending column, which can be added into the  set $\mathcal{S}$ and reduce the residual.

\begin{algorithm}[t]
	\caption{Non-negative OMP (NNOMP)}
	\begin{algorithmic}[1]\label{NNOMP}
		\REQUIRE $K, \ {\bf A} \triangleq \left\{ {\bf a}_1, {\bf a}_2, \cdots, {\bf a}_n, \cdots, {\bf a}_N  \right\}, \ {\bf y} $
		\STATE \text { Initialize } $k = 0, \  \mathcal{S} = \varnothing$, $\bf x = 0$, ${\bf r}_{0} = \bf y $
		\REPEAT
		\STATE  $ i =  \arg\max_{n}$ $ {\bf a}_n^{ T} {\bf r}_{k} $
		\STATE $\mathcal{S} = \mathcal{S} \cup \left\{ i \right\}$	
		\STATE ${\bf x}_{\mathcal{S}} \leftarrow \arg \min_{{\bf z} \geq  0} \left\Vert {\bf y} - {\bf A}_{\mathcal{S}} {\bf z} \right\Vert_2$, \ ${\bf x}_{\mathcal{S}^{\rm c}} \leftarrow \bf 0$
		\STATE $ {\bf r}_{k + 1} = {\bf y} - {\bf Ax}$
		\STATE $k = k + 1$
		\UNTIL  $ \max \left( {\bf A}^{T} {\bf r}_{k} \right) < 0$ or $\vert \mathcal{S} \vert$ equals to $K$
		\ENSURE  $ {\bf x}$
	\end{algorithmic}
\end{algorithm}

However, in our case, because our coefficient matrix   $\bf A$ is not column-normalized, this sparse recovery problem is intractable. The magnitude of each column varies significantly, which  influences the accuracy of finding the correct column that is best correlated with the residual. The NNOMP algorithm  prefers the large-magnitude columns of coefficient matrix  $\bf A$, resulting in  a low probability to select the correct index of non-zero entries.  For example, although there exists one column that is parallel to the residual, the NNOMP algorithm  will still select the wrong column with large magnitude.

To deal with this problem, we design the WNOMP algorithm, which is shown  in Algorithm \ref{ProbNNOMP}. The WNOMP algorithm  can strike a good balance between the effect of column magnitude  and correlation with the residual by setting the dynamic weight for each selection. The weight of magnitude  should be smaller than that of the correlation. If we define the column-wise normalized coefficient matrix  as
\begin{align} 
	\widehat{\bf A} = \left[{{\bf a}_{1} \over \Vert {\bf a}_{1} \Vert_2}, {{\bf a}_{2} \over \Vert {\bf a}_{2} \Vert_2}, \cdots, {{\bf a}_{n} \over \Vert {\bf a}_{n} \Vert_2}, \cdots, {{\bf a}_{N} \over \Vert {\bf a}_{N} \Vert_2} \right],
\end{align}
where ${\bf a}_{n} \in \mathbb{R}^{M\times 1}$ is the $n_{\rm th}$ column vector of coefficient matrix $\bf A$. The dynamic weight $\lambda_{k}$ can be calculated as
\begin{align}
	\lambda_{k} = {\Vert \widehat{\bf A}^T {\bf r}_{k} \Vert_2 \over \sum_{n = 1}^{N} \Vert {\bf a}_n \Vert_2}.
\end{align}

The dynamic weight $\lambda_{k}$ plays an important role in the selection of Step 4 of Algorithm \ref{ProbNNOMP}. An intuitive idea to reduce the influence of column is to use the column-normalized matrix $\widehat{\bf A}$, i.e.,  the 
selection in Step 4 of Algorithm \ref{ProbNNOMP} is modified as $ i =  \arg\max_{n} \left(\left( { {\bf a}_n / \|  {\bf a}_n \|_2} \right)^{\rm T} {\bf r}_{k} \right)$.  However, this modification is not good, because the propagation paths are more likely to exist in the angles with large antenna gain of the transmit antenna array in the base station, which means the large-magnitude column corresponds to  the non-zero entries with a large probability.  As for the selection  $ i =  \arg\max_{n} \left( { {\bf a}_n }^{\rm T} {\bf r}_{k} \right)$, it can be viewed as the multiplication of the modification case and $\|  {\bf a}_n \|_2$, where the correlation with residual contributes little to the selection because  the value of $\|  {\bf a}_n \|_2$ is much larger than $\left( { {\bf a}_n / \|  {\bf a}_n \|_2} \right)^{\rm T} {\bf r}_{k} $.  As a result, the WNOMP considers these two terms separately, and use $\lambda_{k}$ to balance them. The Step 7 of Algorithm \ref{ProbNNOMP} compresses the set $\mathcal{S}$ by updating it with the support of $\bf x$, because some components of  $\mathcal{S}$ may vanish due to the activation of the non-negative constraints.

\begin{algorithm}[t]
	\caption{Weighted non-negative OMP (WNOMP)}
	\begin{algorithmic}[1]\label{ProbNNOMP}
		\REQUIRE $K, \ {\bf A} \triangleq \left\{ {\bf a}_1, {\bf a}_2, \cdots, {\bf a}_n, \cdots, {\bf a}_N  \right\}, \ \widehat{\bf A}, \ {\bf y} $
		\STATE \text { Initialize } $k = 0, \  \mathcal{S} = \varnothing$, $\bf x = 0$, ${\bf r}_{0} = \bf y $
		\REPEAT
		\STATE  $\lambda_{k} = { \left( \Vert \widehat{\bf A}^T {\bf r}_{k} \Vert_2  \right) /  \left( \sum_{n = 1}^{N} \Vert {\bf a}_n \Vert_2 \right)}$
		\STATE  $ i =  \arg\max_{n}$ $\left(\left( { {\bf a}_n / \|  {\bf a}_n \|_2} \right)^{\rm T} {\bf r}_{k} + \lambda_{k} \|  {\bf a}_n \|_2 \right)$
		\STATE $\mathcal{S} = \mathcal{S} \cup \left\{ i \right\}$	
		\STATE ${\bf x}_{\mathcal{S}} \leftarrow \arg \min_{{\bf z} \geq  0} \left\Vert {\bf y} - {\bf A}_{\mathcal{S}} {\bf z} \right\Vert_2$, \ ${\bf x}_{\mathcal{S}^{\rm c}} \leftarrow \bf 0$
	    \STATE  $\mathcal{S} = \operatorname{supp}\left( \bf x\right)$
		\STATE $ {\bf r}_{k + 1} = {\bf y} - {\bf Ax}$
		\STATE $k = k + 1$
		\UNTIL  $ \max \left( {\bf A}^{T} {\bf r}_{k} \right) < 0$ or $\vert \mathcal{S} \vert$ equals to $K$
		\ENSURE  $ {\bf x}$
	\end{algorithmic}
\end{algorithm}

\begin{figure}[t]
	\centering
	\includegraphics[width=2.86in, clip, keepaspectratio]{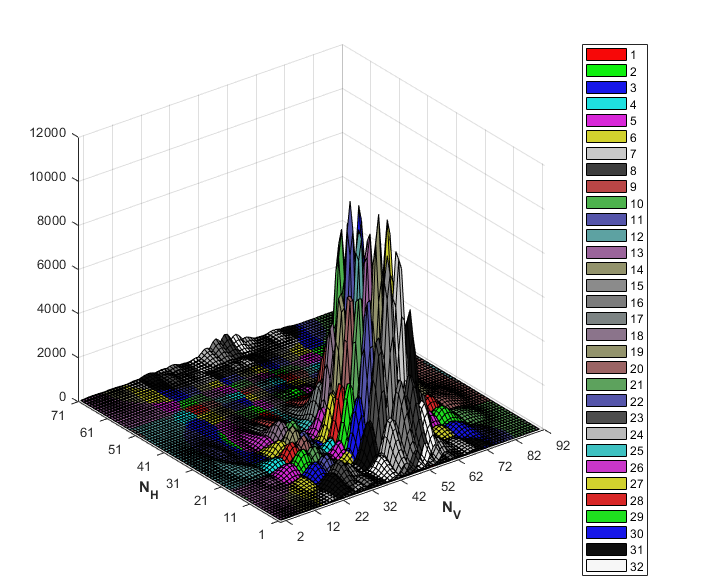}
	\vspace{-5pt}
	\caption{The magnitude of each column of coefficient matrix $\bf A$.}
	\vspace{-5pt}
	\label{fig_2}
\end{figure}

\begin{figure*}[t]
	\centering
	\subfigure[Accuracy versus \# of angular discretization $N$]{\includegraphics[width=2.34 in]{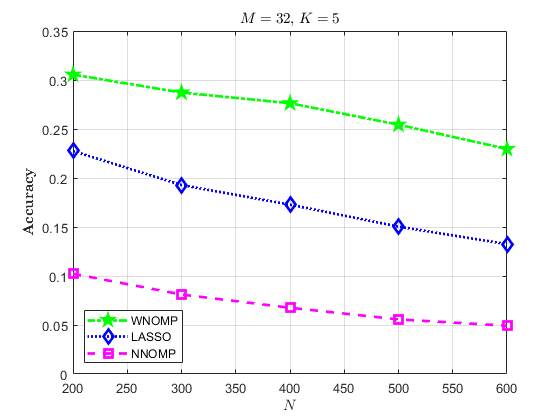}}
	\subfigure[Accuracy versus \# of measured beams $M$]{\includegraphics[width=2.34 in]{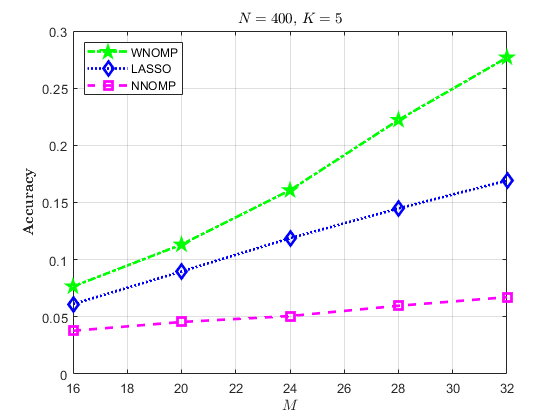}}
	\subfigure[Accuracy versus \# of non-zero entries $K$]{\includegraphics[width=2.34 in]{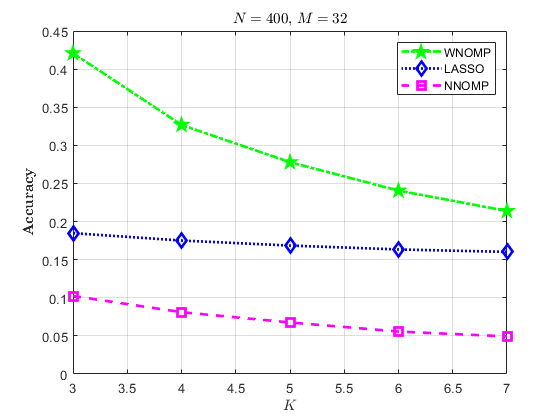}}
	\vspace{-2pt}
	\caption{Simulation results of synthetic data with respect to the performance of accuracy.}
	\vspace{-2pt}
	\label{fig_000}
\end{figure*}

\section{Simulation Results and Discussions}
\label{Section-Results}

In the following, we show numerical results to analyze the performance of  the  algorithms proposed in this paper.  To generate $\bf A$,
 we  set the interval of the tilt angle to be two degrees, and  the interval of the azimuth angle to be five degrees, as shown in  Fig.~\ref{fig_2}. Due to the  antenna gain  $g_{i,j}$, the magnitude (i.e., $l_2$ norm) of each column in $\bf A$ varies significantly, making problem~\eqref{problem2} intractable.  Both synthetic  and real-world measurement data are considered.


%
%
%

For synthetic data,   the location of non-zero entries of synthetic data are generated following the uniform distribution in the interval  $\left[ 1, \ N\right]$.  
The accuracy is defined as the ratio of the number of correctly recovered non-zero entries to the total number of non-zero entries, i.e., the success of support recovery.  Besides the NNOMP and the proposed WNOMP  for the $l_0$ norm constraint, we also consider the least absolute shrinkage and selection operator (LASSO) for sparse recovery, which uses $l_1$ regularization and can be solved by proximal gradient descent.
The performance of the proposed algorithm is related to the problem size, i.e., $N$, $M$ and $K$.  As shown in Fig.~\ref{fig_000}(a), we fix $M$ = 32 and $K$ = 5,  the accuracy decreases when $N$ increases because the algorithms may select the wrong column with more choices. Here, we change the value of $N$ by picking the columns with the top $N$ highest magnitude. The NNOMP  performs worst, while WNOMP  performs  better than LASSO. As shown in Fig.~\ref{fig_000}(b), the accuracy  increases  when $M$ increases for $N$ = 400 and $K$ = 5, which indicates that the channel model is more accurate with more beam measurements.   As shown in Fig.~\ref{fig_000}(c), the accuracy decreases when $K$ increases for $N = 400$ and $M = 32$.

\begin{figure}[t]
	\centering
	\includegraphics[width=2.52in, clip, keepaspectratio]{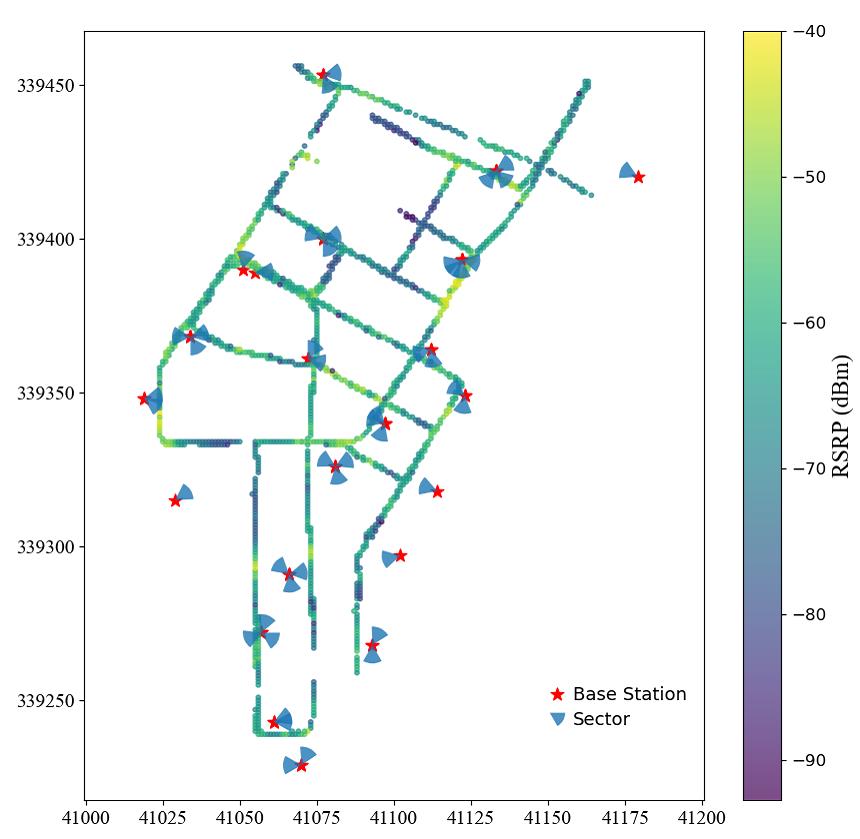}
	\vspace{-2pt}
	\caption{RSRP distribution  in real-world drive testing.}
	\vspace{-2pt}
	\label{fig_22}
\end{figure}

 As for the real-world measurement data, they are collected by using the drive-test in street of the  Chengdu city, China, as shown in Fig.~\ref{fig_22}. The carrier frequency  is $2.6$ GHz with $100$ MHz  bandwidth. There are $N_T = 32$ transmit antennas at the base station.  Depending on the velocity of the driving car,  dozens of samples are collected when passing through a $10 \times 10$ meters grid in the interval of several seconds.  For each grid,  the user equipment can measure the RSRP of the CSI-RS beams (32 beams in total) and SSB beams (8 beams in total) transmitted  from its serving cell and other neighborhood cells. In general, the number of the measured beams in the serving cell is more than the number of  the measured beams in the neighborhood cells. Based on ${rsrp}_{m}(t)$,  we obtain the  RSRP by averaging the samples, then use  them and the coefficient matrix $\bf A$ to formulate and solve the problem~\eqref{problem2}.

{\small
	{\footnotesize  
		\begin{table}[t]
			\vspace{-1pt}
			\caption{MAE performance of predicting the SSB beams \newline (LASSO $\rightarrow$ NNOMP  $\rightarrow$ WNOMP)}
			\vspace{-1pt}
			\centering
			\begin{tabular}{!{\vrule width1.2pt}l!{\vrule width1.2pt}c!{\vrule width1.2pt}c!{\vrule width1.2pt}}  
				\hline 
				& \textbf  {\# of grids}   &  \textbf {MAE (dB)}  \\
				\hline
				{Serving cell}   & 718 & 6.14 $\rightarrow$ 4.45 $\rightarrow$ 4.93 \\
				\hline
				{Neighborhood }   & 1566 & 7.32 $\rightarrow$ 5.85 $\rightarrow$ 5.58 \\
				\hline 
				All grids & 2284 & {\bf \textcolor{red}{6.95}} $\rightarrow$ {\bf \textcolor{red}{5.41}} $\rightarrow$ {\bf \textcolor{red}{5.38}}\\
				\hline
			\end{tabular}
			\label{parameter333}
	\end{table}}\vspace{-1pt}
	
	{\footnotesize 
		\begin{table}[t]
			\vspace{-1pt}
			\caption{MAE performance of predicting the CSI-RS beams \newline (LASSO $\rightarrow$ NNOMP  $\rightarrow$ WNOMP)}
			\vspace{-1pt}
			\centering
			\begin{tabular}{!{\vrule width1.2pt}l!{\vrule width1.2pt}c!{\vrule width1.2pt}c!{\vrule width1.2pt}}  
				\hline 
				& \textbf  {\# of grids}   &  \textbf {MAE (dB)}  \\
				\hline
				{Serving cell}  & 631 & 6.07  $\rightarrow$ 5.37 $\rightarrow$  4.89\\ 
				\hline
				{Neighborhood }  & 1055 & 11.5 $\rightarrow$ 10.6 $\rightarrow$ 8.96 \\
				\hline 
				All grids & 1686 & {\bf \textcolor{red}{9.44}} $\rightarrow$ {\bf \textcolor{red}{8.68}} $\rightarrow$ {\bf \textcolor{red}{7.44}}  \\
				\hline
			\end{tabular}
			\label{parameter444} 
	\end{table}}    
}\vspace{-3pt}

In particular, to evaluate the proposed channel modeling method in practical systems, we collect the RSRP measurements both before and after rotating the antenna array for certain degrees.  The former RSRP data is used for channel modeling by solving  \eqref{problem2} and the solution denoted by $\widehat{\bf x}$. Then we construct the coefficient matrix for the array after rotation, denoted by $\widehat{\bf A}$. The estimated RSRP after array rotation is thus given by $\widehat{\bf y} = \widehat{\bf A}\widehat{\bf x}$. Suppose the RSRP  after rotating as $\widetilde{\bf y}$, we then use the following mean absolute error (MAE)
\begin{align}
	 {\rm MAE} = {1 \over M}\| \widehat{\bf y} - \widetilde{\bf y}  \|_1
\end{align}
as the performance measure.

As shown in Table~\ref{parameter333} of predicting SSB beams, we consider 2284 grids in total and calculate the average MAE.  The MAE of serving cell is better than that of neighborhood cells due to more beam measurements in serving cell. 
The MAE from 1686 grids of predicting CSI-RS beam is shown in  Table~\ref{parameter444}, which is  worse than that of the prediction of  SSB beam, because the beam pattern of CSI-RS is more complex. Although LASSO performs better than the NNOMP in the synthetic data, its performance deteriorates in the real-world measurement data, where the paths does not follow the uniform distribution.
 The MAE performance of the proposed WNOMP is better than those of LASSO and NNOMP algorithms, which means that the channel model established by WNOMP is more accurate than the channel models established by LASSO and NNOMP.

\section{Conclusions}
\label{Section-Conclusions}
In this paper, we propose the LSCM to sense the physical environment that is localized for the specific geographical location.  The proposed LSCM can generate a high-resolution multi-path channel in the angular domain by exploiting the first-order statistics of RSRP measured from multiple  beams. The statistical distribution of the APS in the  LSCM is consistent with that of  true wireless propagation environment. To construct the  LSCM, we design the WNOMP algorithm to deal with the coefficient matrix without column-wise normalization  at a  low computational complexity.  Simulation results verify  the WNOMP  algorithm outperforms the classic LASSO and NNOMP algorithms in terms of accuracy and MAE. Our future work will focus on the joint estimation of the APS by taking advantage of the spatial consistency  in multiple grids \cite{xinzhining}.

{
\appendices
\section{Proof of Theorem~\ref{theorem1} } \label{appendix1}

We want to  establish the relationship between $\mathbb{E}\left( {rsrp}_{m}(t) \right)$ and $ \mathbb{E}\left( {\alpha_{i, j} }(t)  \right)$. We have
\begin{align} 
	\mathbb{E}\left( {rsrp}_{m}(t) \right) &= \mathbb{E}_{\alpha_{i, j} (t) , \omega_{i,j} (t), \omega_{x,y} (t)}\left( {rsrp}_{m}(t) \right) \nonumber\\
	&= \mathbb{E}_{\alpha_{i, j} (t) }\left( \mathbb{E}_{\omega_{i,j} (t), \omega_{x,y} (t)} \left( {rsrp}_{m}(t) \right) \right)\nonumber\\
	&= \mathbb{E}_{\alpha_{i, j} (t) }\left(  {RSRP}_{m}(t) \right), \label{equation1886}
\end{align}
where ${RSRP}_{m}(t)$ is the expectation of ${rsrp}_{m}(t)$ w.r.t. both $\omega_{x,y} (t)$ and $\omega_{i,j} (t)$.  Then, we need to show how ${RSRP}_{m}(t)$ is expressed by $\alpha_{i, j} (t)$. According to the definition in \eqref{eq: rsrp},
{ \small 
	\begin{align}  \nonumber
		& {rsrp}_{m}(t)  = P\left|\sum_{x, y} h_{x, y} (t) w_{x, y}^{(m)}\right|^{2} \\ 
		& = P \! \left( \! \sum_{i = 1}^{N_V} \sum_{j = 1}^{N_H}  \sqrt{\alpha_{i, j} (t)}  g_{i, j} \sum_{x, y}   \cos\left( \! \psi^{(m)}_{i,j,x,y} + \left( \! {\omega_{i,j} (t)} + {\omega_{x,y} (t)} \! \right) \! \right) \! \right)^2 \nonumber \\ &\ \quad + P\left(  \sum_{i = 1}^{N_V} \sum_{j = 1}^{N_H} \sqrt{\alpha_{i, j} (t)}  g_{i, j} \sum_{x, y}   \sin\left( \psi^{(m)}_{i,j,x,y} + \left(   {\omega_{i,j} (t)} + {\omega_{x,y} (t)} \!  \right) \!  \right) \! \right)^2, 
		\label{eq: rsrp22}
	\end{align}
}where $\psi^{(m)}_{i,j,x,y} = 2\pi {d_x x \over \lambda}\cos \theta_i \sin \varphi_j + 2\pi {d_y y \over \lambda  } \sin \theta_i - \phi_{x,y}^{(m)}$. 

The expectation of first term and second term in \eqref{eq: rsrp22} w.r.t. ${\omega_{i,j} (t)}$ and $ {\omega_{x,y} (t)}$ are equal. Then,  we have
{\small
	\begin{align}\label{1789}
		&{RSRP}_{m}(t) \nonumber \\
		&= \mathbb{E}_{\omega_{i,j} (t), \ \omega_{x,y} (t)} \left( {rsrp}_{m}(t) \right) \\
		&= 2P\sum_{i=1}^{N_{V}} \sum_{j=1}^{N_{H}} \sum_{i^{\prime}=1}^{N_{V}} \sum_{j^{\prime}=1}^{N_{H}} \sqrt{\alpha_{i, j} (t)} \sqrt{\alpha_{i^{\prime}}, j^{\prime}(t)}  g_{i, j} g_{i^{\prime}, j^{\prime}} \nonumber \\ & \qquad \qquad \times \mathbb{E} \Bigg[ \sum_{x, y} \cos   \left(\psi_{i, j, x, y}^{(m)} + \omega_{i, j} + \omega_{x, y} \right)
		 \nonumber \\ & \qquad \qquad \qquad \  \sum_{x^{\prime}, y^{\prime}} \cos \left(\psi_{i^{\prime}, j^{\prime}, x^{\prime}, y^{\prime}}^{(m)}+\omega_{i^{\prime}, j^{\prime}} +\omega_{x^{\prime}, y^{\prime}} \right)\Bigg]  \label{1800}.	
	\end{align}
}	
		
It is inferred from the distribution of $\omega_{i, j} $ that 
{\small
	\begin{align}
		&\mathbb{E}\Bigg[\sum_{x, y} \cos \left(\psi_{i, j, x, y}^{(m)}+\left(\omega_{i, j} +\omega_{x, y} \right)\right)   \nonumber \\ &\quad 
		 \sum_{x^{\prime}, y^{\prime}} \cos  \left(\psi_{i^{\prime}, j^{\prime}, x^{\prime}, y^{\prime}}^{(m)}+\left(\omega_{i^{\prime}, j^{\prime}} +\omega_{x^{\prime}, y^{\prime}} \right)\right)\Bigg]=0, 
	\end{align}
}
when $(i, j) \neq\left(i^{\prime}, j^{\prime}\right)$. As a result,	

{\small
	\begin{align}\label{17}
		&{RSRP}_{m}(t) \nonumber \\
		&= 2P \sum_{i=1}^{N_{V}} \sum_{j=1}^{N_{H}} \alpha_{i, j} (t) g_{i, j}^{2} \mathbb{E}\left[\left(\sum_{x, y} \cos \left(\psi_{i, j, x, y}^{(m)}+\omega_{i, j} +\omega_{x, y}\! \right)\!\right)^{2}\!\right] \\
		&= P \sum_{i=1}^{N_{V}} \sum_{j=1}^{N_{H}} \alpha_{i, j} (t) g_{i, j}^{2} \nonumber\\
		& \qquad \times \sum_{x, y} \sum_{x^{\prime}, y^{\prime}}  \mathbb{E}\left[\cos \left(\psi_{i, j, x, y}^{(m)}+\omega_{x, y} -\psi_{i, j, x^{\prime}, y^{\prime}}^{(m)}-\omega_{x^{\prime}, y^{\prime}} \right)\right]. \label{168}	
	\end{align}
}


At last, we can obtain
{\small
	\begin{align}
		& {RSRP}_{m}(t) \nonumber \\
		&= P   \sum_{i = 1}^{N_V} \sum_{j = 1}^{N_H} {\alpha_{i, j} }(t)  g_{i, j}^{2} \Bigg( N_{x} N_{y}\left( 1-e^{-\sigma^{2}} \right) \nonumber \\
		& \qquad \ + e^{-\sigma^{2}} \sum_{x,y}  \sum_{x^{\prime},y^{\prime}}  \cos \left(\psi_{i, j, x, y}^{(m)}-\psi_{i, j, x^{\prime}, y^{\prime}}^{(m)} \right) \Bigg) \label{eq: core} 
		\\ 
		&= \sum_{i = 1}^{N_V} \sum_{j = 1}^{N_H}  {\rm A}_{i, j}^{(m)}  {\alpha_{i, j} }(t) ,
	\end{align}}
where
{	\begin{align} \label{eq: 8}
		{\rm A}_{i, j}^{(m)}   & \triangleq P g_{i, j}^{2}\Bigg( N_{x} N_{y}\left(1-e^{-\sigma^{2}}\right) \nonumber \\ 
		& \qquad \qquad \ + e^{-\sigma^{2}} \sum_{x,y}  \sum_{x^{\prime},y^{\prime}} \cos \left(\psi_{i, j, x, y}^{(m)}-\psi_{i, j, x^{\prime}, y^{\prime}}^{(m)} \right) \Bigg).
	\end{align}
}

}
\flushend
\bibliography{References}
\vfill
\end{document}